\documentclass[sigconf]{acmart}

\copyrightyear{2026}
\acmYear{2026}
\setcopyright{cc}
\setcctype{by}
\acmConference[SA Posters '26]{SIGGRAPH Asia 2026 Posters}{December 01--04, 2026}{Kuala Lumpur, Malaysia}
\acmBooktitle{SIGGRAPH Asia 2026 Posters (SA Posters '26), December 01--04, 2026, Kuala Lumpur, Malaysia}
\acmDOI{10.1145/3829333.3847989}
\acmISBN{979-8-4007-2839-6/2026/12}

\begin{document}

\title{A Neural Hierarchical-Matrix Preconditioner for Real-Time GPU Solves}

\author{Carl Osborne}
\orcid{0009-0008-8753-6612}
\affiliation{%
  \institution{MIT CSAIL}
  \city{Cambridge}
  \country{USA}}
\email{osbo@mit.edu}

\author{Minghao Guo}
\orcid{0000-0003-3408-4997}
\affiliation{%
  \institution{MIT CSAIL}
  \city{Cambridge}
  \country{USA}}
\email{guomh2014@gmail.com}

\author{Crystal Owens}
\orcid{0000-0002-2433-7025}
\affiliation{%
  \institution{MIT CSAIL}
  \city{Cambridge}
  \country{USA}}
\email{crystalo@mit.edu}

\author{Wojciech Matusik}
\orcid{0000-0003-0212-5643}
\affiliation{%
  \institution{MIT CSAIL}
  \city{Cambridge}
  \country{USA}}
\email{wojciech@csail.mit.edu}

\renewcommand{\shortauthors}{Osborne et al.}

\begin{abstract}
Interactive simulation solves $Ax=b$ for a sparse SPD $A$ that changes
every frame, inside an 8--16\,ms budget. At a few thousand unknowns,
algebraic multigrid's setup alone exceeds that budget, while Jacobi and
other local preconditioners have no setup but cannot move error across
the domain. We learn a preconditioner for this gap: a graph-and-attention
network predicts an SPD approximate inverse in $\mathcal{H}^2$-matrix
format. On a spatially ordered 3D mesh, blocks of the true inverse lose
rank as the clusters they couple move apart; the format's nested bases
follow that decay, so inference and apply are dominated by leaf-block
work linear in $N$, where a dense inverse costs $N^2$. Our main finding
concerns training. Probe losses reach $M$ only through a product with
$A$, so their gradient vanishes on the near-null modes that set the
conjugate-gradient iteration count. A truncated Kaporin condition number
has no such factor; changing only the objective cuts iterations on a
held-out frame from 116 to 33. On a ladder of stiff tetrahedral
diffusion problems ours alone fits an 8.3\,ms (120\,fps) frame from
$N=572$ to $3{,}647$.
\end{abstract}



\begin{teaserfigure}
  \centering
  \includegraphics[height=1.580in]{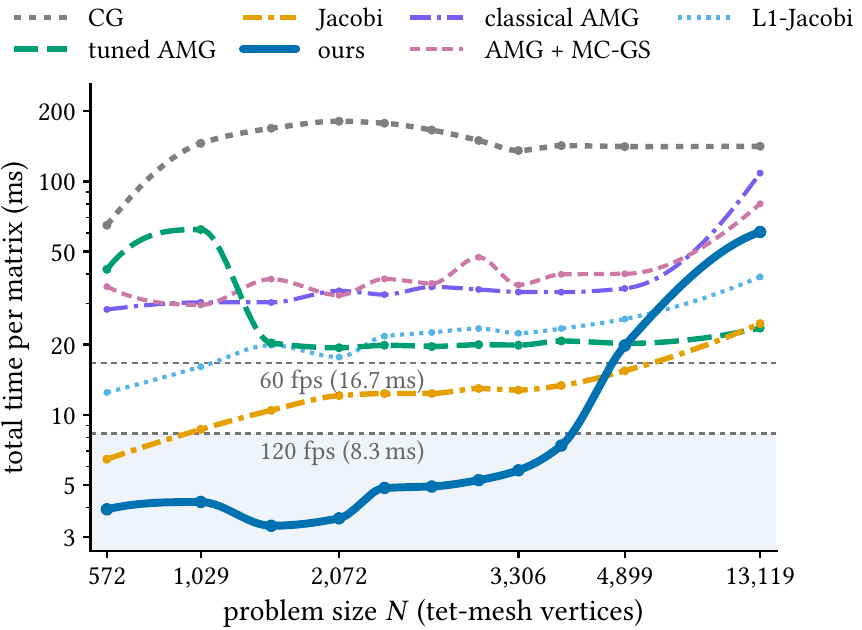}\hfill
  \includegraphics[height=1.371in]{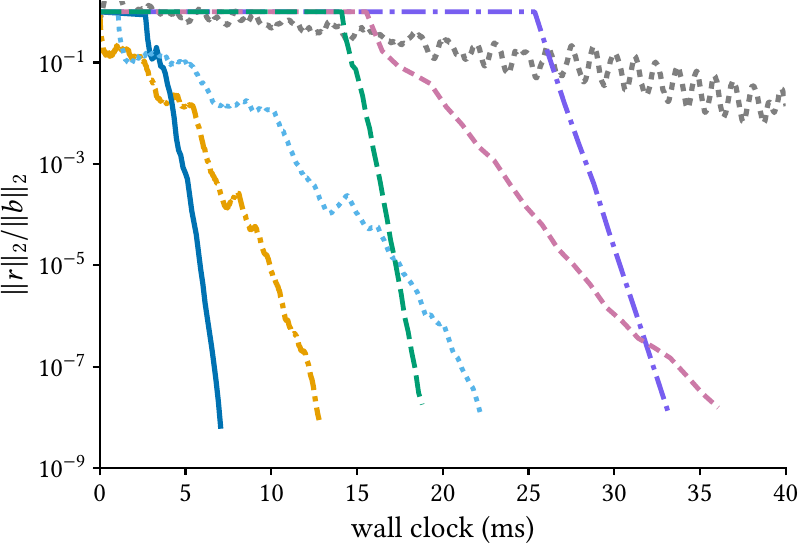}\hfill
  \raisebox{0.191in}{\includegraphics[height=1.275in]{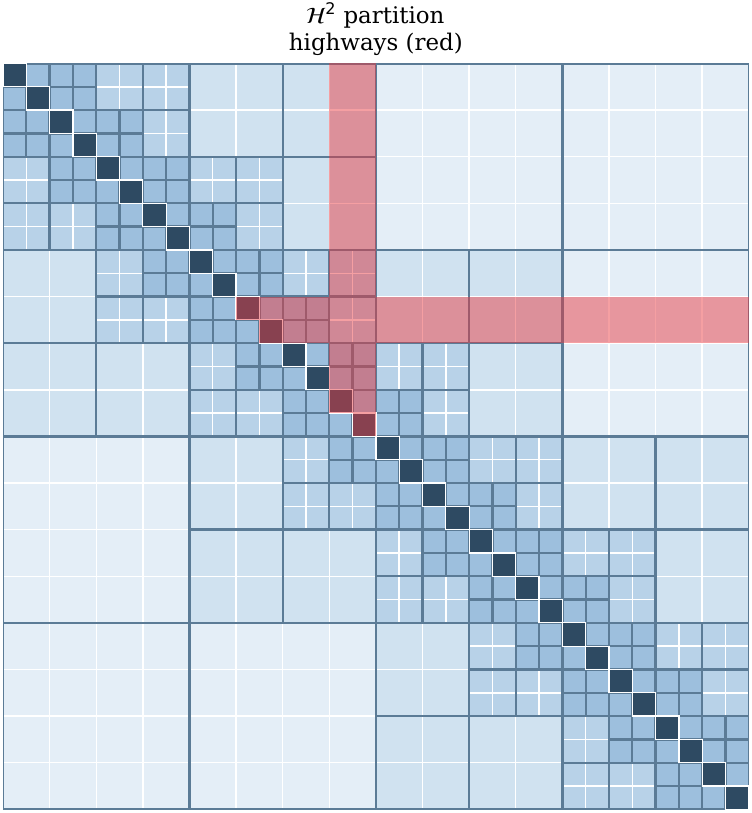}}\hfill
  \raisebox{0.155in}{\includegraphics[height=1.382in]{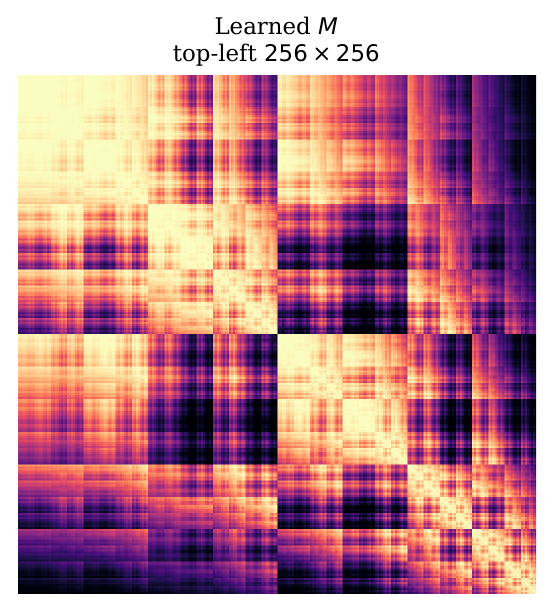}}
  \caption{Left to right: median wall clock per matrix on the
  barrier-heat ladder, where ours alone fits 120\,fps up to ${\sim}3.6$k;
  residual against wall clock at $N=3{,}647$, same legend; the nested
  $\mathcal{H}^2$ partition (highways in red); and a learned factor on a
  2D multiphase Poisson frame.}
  \Description{Four panels. Total solve time against problem size for seven methods,
  with our curve the only one below the 8.3 millisecond line from 572
  to 3,647 unknowns; relative residual against wall-clock time, where
  ours descends fastest; a schematic block partition with dense
  diagonal blocks, low-rank off-diagonal tiles and a highlighted row
  and column strip; and a heat map of the learned preconditioner, whose
  bright square leaf blocks sit on a diagonal against a textured
  off-diagonal background.}
  \label{fig:regime}
\end{teaserfigure}

\maketitle

\begin{figure*}
  \centering
  \includegraphics[width=0.91\textwidth]{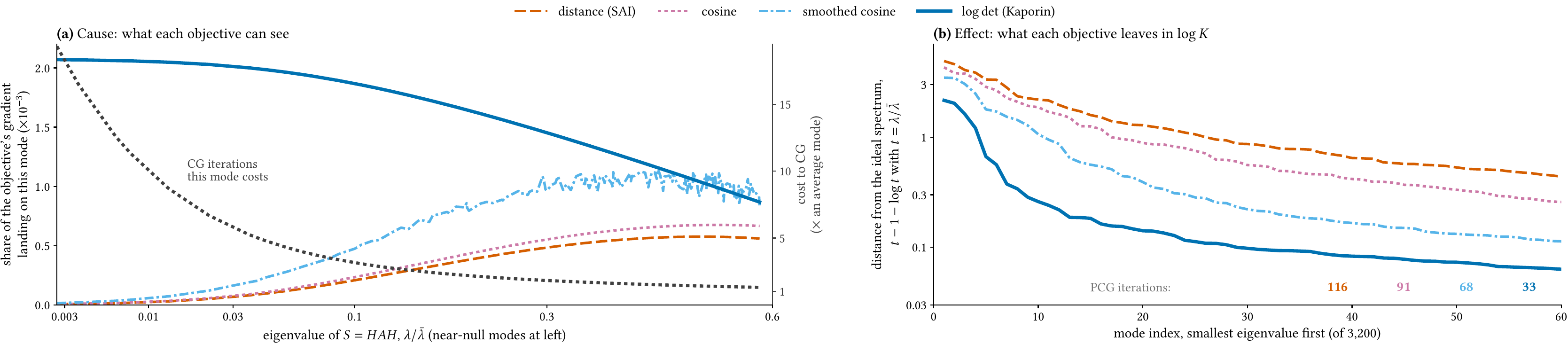}
  \caption{\emph{(a)} What each objective can see on the spectrum of
  $S=HAH$, held-out $N=3{,}306$ frame. A probe loss reaches $M$ only
  through $MAz$, so its gradient share on a mode (left axis) carries the
  eigenvalue and vanishes at the near-null end, where CG spends its
  iterations (right axis). \emph{(b)} What each objective leaves in
  $\log K$: each mode's term after training, with its PCG count.}
  \Description{Two panels. Left: gradient share per mode against
  eigenvalue on a linear vertical axis, where the distance and cosine
  curves sit at zero across the small eigenvalues while the
  log-determinant curve stays high, together with a dotted curve on a
  second axis showing CG's cost per mode rising towards the small
  eigenvalues. Right: each of the sixty smallest modes' term in log K
  after training with each objective on a logarithmic axis, largest at
  the smallest modes and lowest throughout for the Kaporin-trained model,
  with each objective's PCG iteration count printed inside the panel.}
  \label{fig:objective}
\end{figure*}

\section{Setting}
Interactive simulations of pressure, implicit diffusion and elasticity
solve $Ax=b$ with $A$ sparse, SPD and different every frame, and at
60--120\,fps the whole solve gets 8--16\,ms on the GPU. The systems we
target are small by solver standards and sit inside larger pipelines:
the coarse levels of a multigrid cycle, the per-subdomain solves of a
Schwarz preconditioner \cite{wu2022mas}, and the systems that follow a
change of contact set.

Two families of preconditioner meet this regime from opposite sides.
Algebraic multigrid \cite{naumov2015amgx} moves error between scales,
so its iteration count barely grows with $N$; but on an unstructured
mesh it rebuilds its hierarchy for every new matrix, and that setup
alone costs 12--16\,ms at our sizes. Jacobi, and learned sparse factors
and approximate inverses \cite{Li2023learning,Yang2025sparse}, have
negligible setup but act locally, so they stall when error has to cross
long, high-contrast seams. We want the global transport of multigrid at
the setup cost of a local method: setup is one feed-forward inference
(1.8--2.5\,ms). Hierarchical matrices
give the operator its shape; networks built on that structure have so
far learned solution maps with the hierarchy fixed in their weights
\cite{Fan2019nested}, or compressed a learned Green's function into an
$\mathcal{H}$-matrix per operator \cite{Xu2025neural}. We predict the
factor for each new matrix and use it inside conjugate gradients.

\section{Preconditioner}
\paragraph{Structure and apply.}
Unknowns are ordered along a space-filling curve and cut into $K$
\emph{leaves} of $L=128$ nodes. This induces a
weak-admissibility $\mathcal{H}$-matrix partition (Fig.~\ref{fig:regime},
third panel): $K$ dense diagonal blocks and $O(K)$ tiles that grow
away from the diagonal.

The partition follows the inverse. For elliptic operators, even with
rough, high-contrast coefficients like our barriers, a block of $A^{-1}$
coupling two separated clusters is numerically
low-rank~\cite{bebendorf2003existence}, and for neighbouring clusters,
which weak admissibility also compresses, the rank grows with the area
of their shared interface rather than their volume. We subdivide each
tile by this interface-area rule: a tile spanning $S$ leaves splits into
sub-blocks of about $S^{1/3}$ leaves, so its rank grows as $S^{2/3}$,
and every sub-block touching a leaf reuses that leaf's basis. These
nested bases make an $\mathcal{H}$-matrix
an $\mathcal{H}^2$-matrix~\cite{hackbusch2000h2}. Storage and apply
are dominated by the $NL$ entries of the leaf blocks, linear in $N$ at
fixed $L$, against a dense inverse's $N^2$. Applying $M$ is a batch of
dense matrix products; with no host allocation or data-dependent
control flow, each PCG iteration is one CUDA graph.

\paragraph{Network.}
A width-128 encoder (an MLP and two $A$-weighted graph convolutions)
feeds two transformer stacks over the block layout. A diagonal stack
emits each dense leaf block as $F_k=U_kU_k^\top$; an off-diagonal stack
emits low-rank factor pairs in the shared bases from strip embeddings
pooled to 32 tokens per tile of any size. Scatter-added
row, column and global buffers (``highways'') route information between
blocks at $O(Nd)$: a communication scaffold, not a coarse solve. A
learned per-node Jacobi gate adds $\mathrm{diag}(\boldsymbol\lambda)D^{-1/2}$
to $H$, helping early in training. Inference is $\Theta(N)$ in the leaf
products plus two attention launches over $K=N/L$ tokens.

\paragraph{SPD by construction.}
Mirroring the off-diagonal blocks makes $H$ symmetric, so $M=H^2$ is SPD
whenever $H$ is nonsingular; we monitor $\lambda_{\min}(H)$ and would
fall back to Jacobi, and no held-out frame has needed it. Learning
$M\approx A^{-1}$ directly gave indefinite $M$ on unseen frames, which
is what motivated the squared form.

\section{Training objective}
\paragraph{What a probe loss can see.}
Let $S=HAH$, similar to $MA$. Conjugate gradients spends its iterations
on the smallest eigenvalues of $S$ (Fig.~\ref{fig:objective}a, right
axis), the modes $M$ must amplify most: on an eigenmode $Aw=\lambda w$
the ideal $M$ acts as $Mw=w/\lambda$. Approximate inverses are trained self-supervised on
probes: draw $z$, form $Az$, and compare $MAz$ with $z$, whether by the
distance $\|MAz-z\|_2^2$ with $A$ scaled by $\|A\|_\infty$ (the
sparse-approximate-inverse loss of \citet{Yang2025sparse} up to norm and
operand order), or by an angular loss on a block $Z$ of $n_z=256$
probes,
$\mathcal{L}_{\cos}=1-\langle MAZ,Z\rangle_F/(\|MAZ\|_F\|Z\|_F)$. The
distance asks the spectrum of $MA$ to cluster at the particular value
$\|A\|_\infty$; the angular loss is blind to the scale of $M$, so it
rewards clustering wherever the spectrum can gather most tightly.

Every one of them touches $M$ only through its product with $A$. With
$y=MAz$,
$\partial\mathcal{L}/\partial M=(\partial\mathcal{L}/\partial y)(Az)^\top$,
which on a mode $Aw=\lambda w$ is $O(\lambda)$; the $AMz$ ordering picks
up the same factor through $A^\top$. $A$
annihilates the mode before $M$ can learn from it, and the signal that
would teach $M$ its required action is suppressed by the very factor
that makes the mode expensive (Fig.~\ref{fig:objective}a). The symptom
is known: \citet{Chen2025graph} reweights probes toward the bottom of
the spectrum, and \citet{Xu2025neural} report a flat loss surface near
singular modes. Reweighting moves probe energy, not the factor.
Pre-smoothing, $Z\leftarrow(I-\omega D^{-1/2}A)Z$ twice with
$\omega=0.6$ and $D=\mathrm{diag}(A)$, the relaxation bootstrap AMG
uses to expose algebraically smooth error
\cite{brandt2011bootstrap}, shifts probe energy toward the low end and
cuts iterations from 91 to 68: a tilt, not a fix.

\paragraph{A term that never passes a probe through $A$.}
One quantity escapes this. The gradient of $-\log\det MA$ with respect
to $M$ is $-M^{-\top}$: no probe, no product with $A$, and the same
weight on every mode in $\log\lambda$, the near-null modes included.

Kaporin's condition number makes that term a training target:
\begin{equation}
  K(S)=\frac{\big(\tfrac{1}{N}\mathrm{tr}\,S\big)^N}{\det S},
\end{equation}
the ratio of the arithmetic to the geometric mean of the eigenvalues of
$S$, raised to the $N$th power. It is 1 exactly when every eigenvalue
is equal, it does not change under $M\to sM$, and it bounds PCG
directly: $\lceil\log_2 K+\log_2(1/\varepsilon)\rceil$ iterations reduce
the $M$-norm residual by $\varepsilon$~\cite{kaporin1994new}.
\citet{cali2023neural} train on $K$ with triangular factors, whose
determinant is cheap. Ours has no such shortcut, and a stochastic
Lanczos estimate of $\log\det$ is too noisy to descend on.

So we truncate. Only the $p=8$ smallest eigenvalues
$\lambda_1\le\cdots\le\lambda_p$ of $S$ are treated exactly; everything
above them is summarized by its mean $\bar m$. We minimize
\begin{equation}
  \widehat{\log K} = N\log\Big(\tfrac{1}{N}\mathrm{tr}\,S^2\Big)
    - 2\Big[\sum_{i\le p}\log\lambda_i
    + (N-p)\log\bar m\Big],
\end{equation}
an estimate of $\log K(S^2)$ whose traces come from Hutchinson probes.
The tail is exact: the $p$ smallest eigenpairs come from
Chebyshev-filtered subspace iteration and are detached, so gradient
reaches them through their Rayleigh quotients, and there the estimator
has the log-determinant's $\lambda$-free gradient exactly. The bulk is
approximate: the mean stands in for its log-determinant, so by Jensen's
inequality the estimate is a lower bound, and the returning $\lambda$
factor turns the term into pressure to cluster the bulk.

\paragraph{What changes.}
With the same architecture and data, PCG iterations on the held-out
$N=3{,}306$ frame fall from 116 (distance) to 91 (cosine), 68 (smoothed
probes) and 33 (truncated Kaporin), and $\kappa(S)$ from 611 to 42.
Even with the suppression gone from training, the low end remains the
hardest part of the spectrum for this format to learn
(Fig.~\ref{fig:objective}b): after Kaporin training, the eight smallest
modes still hold $26\%$ of the $\log K$ that remains.

\section{Results}
\emph{Benchmark.} Heat diffusion on TetWild
tetrahedralizations~\cite{hu2018tetwild} of a Thingi10K
solid~\cite{zhou2016thingi10k} from $N=572$ to $13{,}119$, with a random
barrier field of contrast $10^4$ resampled every frame. One network is
trained per size on 60--100 frames, in about ten minutes on one GPU,
and evaluated on unseen frames of the same mesh. \emph{Baselines.} Everything stays on the GPU: at these sizes a host
round-trip costs more than the solve. We compare unpreconditioned CG, GPU Jacobi,
and tuned AMG, the best of an AmgX~\cite{naumov2015amgx} sweep over
aggregation and classical hierarchies with L1-Jacobi, multicolour
Gauss--Seidel and DILU smoothers. At $N\lesssim1$k the hierarchy
collapses to a one-iteration direct coarse solve, and that dense solve
is what AMG pays for there. \emph{Protocol.} Every solve runs to a true
relative residual of $10^{-8}$.

\begin{table}[!htbp]
  \caption{Per new matrix on one L40S: medians over unseen test frames;
  each of our columns is its own median.}
  \label{tab:ladder}
  \scriptsize
  \begin{tabular}{@{}rrrr rrrr@{}}
    \toprule
    & \multicolumn{3}{c}{Baselines, ms (iters)} & \multicolumn{4}{c}{Ours} \\
    \cmidrule(lr){2-4}\cmidrule(lr){5-8}
    $N$ & CG & Jacobi & AMG & inf & solve & it (min--max) & tot \\
    \midrule
    572      & 64.9 (1938) & 6.4 (192) & 42.0 (1) & 1.9 & 1.6 & 24 (18--39) & \textbf{3.9} \\
    1{,}029  & 145.6 (4269) & 8.7 (252) & 62.2 (1) & 1.8 & 2.3 & 34 (33--42) & \textbf{4.2} \\
    1{,}479  & 168.9 (5000) & 10.4 (306) & 20.3 (28) & 1.9 & 1.5 & 21 (21--45) & \textbf{3.3} \\
    2{,}072  & 181.1 (5000) & 12.1 (332) & 19.4 (26) & 1.8 & 1.8 & 24 (24--27) & \textbf{3.6} \\
    2{,}695  & 166.0 (4530) & 12.3 (336) & 19.6 (23) & 1.9 & 3.0 & 39 (33--45) & \textbf{4.9} \\
    3{,}306  & 135.5 (3636) & 12.7 (342) & 19.9 (23) & 2.4 & 3.3 & 44 (36--51) & \textbf{5.8} \\
    3{,}647  & 142.3 (3738) & 13.3 (348) & 20.7 (24) & 2.5 & 4.8 & 39 (33--48) & \textbf{7.4} \\
    \midrule
    4{,}899  & 140.9 (3598) & \textbf{15.4} (394) & 20.2 (24) & 4.3 & 15.5 & 192 (177--195) & 19.8 \\
    13{,}119 & 141.4 (3460) & 24.6 (582) & \textbf{23.5} (32) & 9.8 & 50.9 & 306 (300--327) & 60.7 \\
    \bottomrule
  \end{tabular}
\end{table}

Through ${\sim}3.6$k the solve alone (1.5--4.8\,ms) is faster than
Jacobi or AMG, and with inference the total stays inside 120\,fps;
beyond that, Jacobi and then AMG overtake us. Residual traces
(Fig.~\ref{fig:regime}, second panel) separate early, so looser
tolerances widen the band.

\section{The learning gap and generalization}
Leaf diagonals are learned first, and the off-diagonal cores degrade
first as $N$ grows. This is a learning gap rather than a missing level
of the hierarchy: on smaller meshes, fitting the same format to
$A^{-1/2}$ directly by per-tile SVD gives an oracle that converges in
3--17 iterations, so the format is not the binding constraint.

Generalization is narrow in our trained example, whose setup is tiny:
one mesh, 60--100 frames and minutes of training per size.
Held-out barrier fields at a trained size stay in band, but a network
trained at $N=2{,}072$ needs 314 iterations at
$3{,}647$, one applied below it hits the cap, and a chamber
geometry at $N=3{,}306$ costs 81 iterations against 44 in
distribution (supplementary).

\paragraph{Limitations.}
Quality lags an exact fit of the format and collapses past ${\sim}5$k,
where that fit degrades mildly. The format and objective can beat
local methods; training a network that does so robustly across sizes,
geometries and PDEs is open.

\bibliographystyle{ACM-Reference-Format}
\bibliography{references}


\begin{thebibliography}{14}


\ifx \showCODEN    \undefined \def \showCODEN     #1{\unskip}     \fi
\ifx \showISBNx    \undefined \def \showISBNx     #1{\unskip}     \fi
\ifx \showISBNxiii \undefined \def \showISBNxiii  #1{\unskip}     \fi
\ifx \showISSN     \undefined \def \showISSN      #1{\unskip}     \fi
\ifx \showLCCN     \undefined \def \showLCCN      #1{\unskip}     \fi
\ifx \shownote     \undefined \def \shownote      #1{#1}          \fi
\ifx \showarticletitle \undefined \def \showarticletitle #1{#1}   \fi
\ifx \showURL      \undefined \def \showURL       {\relax}        \fi
\providecommand\bibfield[2]{#2}
\providecommand\bibinfo[2]{#2}
\providecommand\natexlab[1]{#1}
\providecommand\showeprint[2][]{arXiv:#2}

\bibitem[Bebendorf and Hackbusch(2003)]%
        {bebendorf2003existence}
\bibfield{author}{\bibinfo{person}{Mario Bebendorf} {and}
  \bibinfo{person}{Wolfgang Hackbusch}.} \bibinfo{year}{2003}\natexlab{}.
\newblock \showarticletitle{Existence of $\mathcal{H}$-Matrix Approximants to
  the Inverse {FE}-Matrix of Elliptic Operators with
  {$L^\infty$}-Coefficients}.
\newblock \bibinfo{journal}{\emph{Numer. Math.}} \bibinfo{volume}{95},
  \bibinfo{number}{1} (\bibinfo{year}{2003}), \bibinfo{pages}{1--28}.
\newblock
\href{https://doi.org/10.1007/s00211-002-0445-6}{doi:\nolinkurl{10.1007/s00211-002-0445-6}}


\bibitem[Brandt et~al\mbox{.}(2011)]%
        {brandt2011bootstrap}
\bibfield{author}{\bibinfo{person}{Achi Brandt}, \bibinfo{person}{James
  Brannick}, \bibinfo{person}{Karsten Kahl}, {and} \bibinfo{person}{Ira
  Livshits}.} \bibinfo{year}{2011}\natexlab{}.
\newblock \showarticletitle{Bootstrap {AMG}}.
\newblock \bibinfo{journal}{\emph{SIAM Journal on Scientific Computing}}
  \bibinfo{volume}{33}, \bibinfo{number}{2} (\bibinfo{year}{2011}),
  \bibinfo{pages}{612--632}.
\newblock
\href{https://doi.org/10.1137/090752973}{doi:\nolinkurl{10.1137/090752973}}


\bibitem[Cal{\`i} et~al\mbox{.}(2023)]%
        {cali2023neural}
\bibfield{author}{\bibinfo{person}{Salvatore Cal{\`i}},
  \bibinfo{person}{Daniel~C. Hackett}, \bibinfo{person}{Yin Lin},
  \bibinfo{person}{Phiala~E. Shanahan}, {and} \bibinfo{person}{Brian Xiao}.}
  \bibinfo{year}{2023}\natexlab{}.
\newblock \showarticletitle{Neural-network preconditioners for solving the
  {D}irac equation in lattice gauge theory}.
\newblock \bibinfo{journal}{\emph{Physical Review D}} \bibinfo{volume}{107},
  \bibinfo{number}{3} (\bibinfo{year}{2023}), \bibinfo{pages}{034508}.
\newblock
\href{https://doi.org/10.1103/PhysRevD.107.034508}{doi:\nolinkurl{10.1103/PhysRevD.107.034508}}


\bibitem[Chen(2025)]%
        {Chen2025graph}
\bibfield{author}{\bibinfo{person}{Jie Chen}.} \bibinfo{year}{2025}\natexlab{}.
\newblock \showarticletitle{Graph Neural Preconditioners for Iterative
  Solutions of Sparse Linear Systems}. In
  \bibinfo{booktitle}{\emph{International Conference on Learning
  Representations}}.
\newblock


\bibitem[Fan et~al\mbox{.}(2019)]%
        {Fan2019nested}
\bibfield{author}{\bibinfo{person}{Yuwei Fan}, \bibinfo{person}{Jordi
  Feliu-Fab{\`a}}, \bibinfo{person}{Lin Lin}, \bibinfo{person}{Lexing Ying},
  {and} \bibinfo{person}{Leonardo Zepeda-N{\'u}{\~n}ez}.}
  \bibinfo{year}{2019}\natexlab{}.
\newblock \showarticletitle{A Multiscale Neural Network Based on Hierarchical
  Nested Bases}.
\newblock \bibinfo{journal}{\emph{Research in the Mathematical Sciences}}
  \bibinfo{volume}{6}, \bibinfo{number}{2}, Article \bibinfo{articleno}{21}
  (\bibinfo{year}{2019}).
\newblock
\href{https://doi.org/10.1007/s40687-019-0183-3}{doi:\nolinkurl{10.1007/s40687-019-0183-3}}


\bibitem[Hackbusch et~al\mbox{.}(2000)]%
        {hackbusch2000h2}
\bibfield{author}{\bibinfo{person}{Wolfgang Hackbusch},
  \bibinfo{person}{Boris~N. Khoromskij}, {and} \bibinfo{person}{Stefan~A.
  Sauter}.} \bibinfo{year}{2000}\natexlab{}.
\newblock \showarticletitle{On $\mathcal{H}^2$-Matrices}.
\newblock In \bibinfo{booktitle}{\emph{Lectures on Applied Mathematics}},
  \bibfield{editor}{\bibinfo{person}{Hans-Joachim Bungartz},
  \bibinfo{person}{Ronald H.~W. Hoppe}, {and} \bibinfo{person}{Christoph
  Zenger}} (Eds.). \bibinfo{publisher}{Springer}, \bibinfo{address}{Berlin},
  \bibinfo{pages}{9--29}.
\newblock
\href{https://doi.org/10.1007/978-3-642-59709-1_2}{doi:\nolinkurl{10.1007/978-3-642-59709-1_2}}


\bibitem[Hu et~al\mbox{.}(2018)]%
        {hu2018tetwild}
\bibfield{author}{\bibinfo{person}{Yixin Hu}, \bibinfo{person}{Qingnan Zhou},
  \bibinfo{person}{Xifeng Gao}, \bibinfo{person}{Alec Jacobson},
  \bibinfo{person}{Denis Zorin}, {and} \bibinfo{person}{Daniele Panozzo}.}
  \bibinfo{year}{2018}\natexlab{}.
\newblock \showarticletitle{Tetrahedral Meshing in the Wild}.
\newblock \bibinfo{journal}{\emph{ACM Transactions on Graphics}}
  \bibinfo{volume}{37}, \bibinfo{number}{4} (\bibinfo{year}{2018}),
  \bibinfo{pages}{60:1--60:14}.
\newblock
\href{https://doi.org/10.1145/3197517.3201353}{doi:\nolinkurl{10.1145/3197517.3201353}}


\bibitem[Kaporin(1994)]%
        {kaporin1994new}
\bibfield{author}{\bibinfo{person}{Igor~E. Kaporin}.}
  \bibinfo{year}{1994}\natexlab{}.
\newblock \showarticletitle{New Convergence Results and Preconditioning
  Strategies for the Conjugate Gradient Method}.
\newblock \bibinfo{journal}{\emph{Numerical Linear Algebra with Applications}}
  \bibinfo{volume}{1}, \bibinfo{number}{2} (\bibinfo{year}{1994}),
  \bibinfo{pages}{179--210}.
\newblock
\href{https://doi.org/10.1002/nla.1680010208}{doi:\nolinkurl{10.1002/nla.1680010208}}


\bibitem[Li et~al\mbox{.}(2023)]%
        {Li2023learning}
\bibfield{author}{\bibinfo{person}{Yichen Li}, \bibinfo{person}{Peter~Yichen
  Chen}, \bibinfo{person}{Tao Du}, {and} \bibinfo{person}{Wojciech Matusik}.}
  \bibinfo{year}{2023}\natexlab{}.
\newblock \showarticletitle{Learning Preconditioners for Conjugate Gradient
  {PDE} Solvers}. In \bibinfo{booktitle}{\emph{Proceedings of the 40th
  International Conference on Machine Learning}}
  \emph{(\bibinfo{series}{Proceedings of Machine Learning Research},
  Vol.~\bibinfo{volume}{202})}. \bibinfo{publisher}{PMLR},
  \bibinfo{pages}{19425--19439}.
\newblock


\bibitem[Naumov et~al\mbox{.}(2015)]%
        {naumov2015amgx}
\bibfield{author}{\bibinfo{person}{M. Naumov}, \bibinfo{person}{M. Arsaev},
  \bibinfo{person}{P. Castonguay}, \bibinfo{person}{J. Cohen},
  \bibinfo{person}{J. Demouth}, \bibinfo{person}{J. Eaton}, \bibinfo{person}{S.
  Layton}, \bibinfo{person}{N. Markovskiy}, \bibinfo{person}{I. Reguly},
  \bibinfo{person}{N. Sakharnykh}, \bibinfo{person}{V. Sellappan}, {and}
  \bibinfo{person}{R. Strzodka}.} \bibinfo{year}{2015}\natexlab{}.
\newblock \showarticletitle{{AmgX}: A Library for {GPU} Accelerated Algebraic
  Multigrid and Preconditioned Iterative Methods}.
\newblock \bibinfo{journal}{\emph{SIAM Journal on Scientific Computing}}
  \bibinfo{volume}{37}, \bibinfo{number}{5} (\bibinfo{year}{2015}),
  \bibinfo{pages}{S602--S626}.
\newblock
\href{https://doi.org/10.1137/140980260}{doi:\nolinkurl{10.1137/140980260}}


\bibitem[Wu et~al\mbox{.}(2022)]%
        {wu2022mas}
\bibfield{author}{\bibinfo{person}{Botao Wu}, \bibinfo{person}{Zhendong Wang},
  {and} \bibinfo{person}{Huamin Wang}.} \bibinfo{year}{2022}\natexlab{}.
\newblock \showarticletitle{A {GPU}-Based Multilevel Additive {Schwarz}
  Preconditioner for Cloth and Deformable Body Simulation}.
\newblock \bibinfo{journal}{\emph{ACM Transactions on Graphics}}
  \bibinfo{volume}{41}, \bibinfo{number}{4} (\bibinfo{year}{2022}),
  \bibinfo{pages}{63:1--63:14}.
\newblock
\href{https://doi.org/10.1145/3528223.3530085}{doi:\nolinkurl{10.1145/3528223.3530085}}


\bibitem[Xu et~al\mbox{.}(2025)]%
        {Xu2025neural}
\bibfield{author}{\bibinfo{person}{Tianshi Xu}, \bibinfo{person}{Rui~Peng Li},
  {and} \bibinfo{person}{Yuanzhe Xi}.} \bibinfo{year}{2025}\natexlab{}.
\newblock \showarticletitle{Neural Approximate Inverse Preconditioners}.
\newblock \bibinfo{journal}{\emph{arXiv preprint arXiv:2510.13034}}
  (\bibinfo{year}{2025}).
\newblock


\bibitem[Yang et~al\mbox{.}(2025)]%
        {Yang2025sparse}
\bibfield{author}{\bibinfo{person}{Zherui Yang}, \bibinfo{person}{Zhehao Li},
  \bibinfo{person}{Kangbo Lyu}, \bibinfo{person}{Yixuan Li},
  \bibinfo{person}{Tao Du}, {and} \bibinfo{person}{Ligang Liu}.}
  \bibinfo{year}{2025}\natexlab{}.
\newblock \showarticletitle{Learning Sparse Approximate Inverse Preconditioners
  for Conjugate Gradient Solvers on {GPUs}}. In
  \bibinfo{booktitle}{\emph{Advances in Neural Information Processing
  Systems}}.
\newblock


\bibitem[Zhou and Jacobson(2016)]%
        {zhou2016thingi10k}
\bibfield{author}{\bibinfo{person}{Qingnan Zhou} {and} \bibinfo{person}{Alec
  Jacobson}.} \bibinfo{year}{2016}\natexlab{}.
\newblock \showarticletitle{{Thingi10K}: A Dataset of 10,000 3D-Printing
  Models}.
\newblock \bibinfo{journal}{\emph{arXiv preprint arXiv:1605.04797}}
  (\bibinfo{year}{2016}).
\newblock


\end{thebibliography}

\end{document}